\documentclass[10pt, onecolumn, draftcls]{IEEEtran}

\usepackage{amsmath}
\usepackage{amsfonts}
\usepackage{amssymb}
\usepackage{graphicx}
\usepackage{epsfig}
\usepackage{mathrsfs}
\usepackage{cite}
\usepackage{psfig}
\usepackage{epsf}
\usepackage{flushend}
\usepackage[active]{srcltx}

\newtheorem{theorem}{Theorem}

\newtheorem{lemma}[theorem]{Lemma}
\newtheorem{corol}[theorem]{Corollary}

\renewcommand{\baselinestretch}{1}

\begin{document}

\title{Rate-Constrained Wireless Networks with Fading Channels: Interference-Limited and Noise-Limited Regimes}
\author{Masoud Ebrahimi and Amir K. Khandani
\thanks{This work is financially supported by Nortel Networks and by matching funds from
the federal government of Canada (NSERC) and province of Ontario
(OCE).}
\thanks{The authors are affiliated with the Coding and Signal Transmission Laboratory,
Electrical and Computer Engineering Department, University of
Waterloo, Waterloo, ON, N2L 3G1, Canada, Tel: 519-885-1211 ext.
35324, \mbox{Fax: 519-888-4338}, Emails:~\{masoud,
khandani\}@cst.uwaterloo.ca.}
\thanks{This material was presented in part at the 10th Canadian Workshop on
Information Theory (CWIT), Edmonton, AB, Canada, June 2007.}
}\maketitle \markboth{Submitted to IEEE Transactions on Information
Theory}{}

\begin{abstract}
A network of $n$ wireless communication links is considered in a
Rayleigh fading environment. It is assumed that each link can be
active and transmit with a constant power $P$ or remain silent. The
objective is to maximize the number of active links such that each
active link can transmit with a constant rate $\lambda$. An upper
bound is derived that shows the number of active links scales at
most like $\frac{1}{\lambda} \log n$. To obtain a lower bound, a
decentralized link activation strategy is described and analyzed. It
is shown that for small values of $\lambda$, the number of supported
links by this strategy meets the upper bound; however, as $\lambda$
grows, this number becomes far below the upper bound. To shrink the
gap between the upper bound and the achievability result, a modified
link activation strategy is proposed and analyzed based on some
results from random graph theory. It is shown that this modified
strategy performs very close to the optimum. Specifically, this
strategy is \emph{asymptotically almost surely} optimum when
$\lambda$ approaches $\infty$ or $0$. It turns out the optimality
results are obtained in an interference-limited regime. It is
demonstrated that, by proper selection of the algorithm parameters,
the proposed scheme also allows the network to operate in a
noise-limited regime in which the transmission rates can be adjusted
by the transmission powers. The price for this flexibility is a
decrease in the throughput scaling law by a multiplicative factor of
$\log \log n$.
\end{abstract}

\begin{keywords}
Wireless networks, fading channel, throughput, scaling law, random
graph.
\end{keywords}

\section{Introduction}

Wireless networks consist of a number of nodes communicating over a
shared wireless channel. The design and analysis of such
configurations, even in their simplest forms, are among the most
difficult problems in information theory. However, as the number of
nodes becomes large, wireless networks become more tractable, where
scaling laws for network parameters, such as throughput, can be
derived.

Most of the works dealing with the throughput of large wireless
networks consider a channel model in which the signal power decays
according to a distance-based attenuation law
\cite{Gupta00,Franceschetti07j,Xie04,Gastpar05,Leveque05,Dousse06,Grossglauser02,Kulkarni04}.
However, in a wireless environment, the presence of obstacles and
scatterers adds some randomness to the received signal. This random
behavior of the channel, known as fading, can drastically change the
scaling laws of a network in both multihop
\cite{Gowaikar05,Gowaikar06c,Toumpis04,Xue05} and single-hop
scenarios \cite[Chapter
8]{Etkinthesis},\cite{Weber06c,Weber06,Gesbert07rawnet,Gesbert07,Jindal07}.
In this paper, we follow the model of \cite{Gowaikar05, Etkinthesis,
Ebrahimi07}, where fading is assumed to be the dominant factor
affecting the strength of the channels between nodes.

In this work, we consider a single-hop network, i.e., a network in
which data is transmitted directly from sources to their
corresponding receivers without utilizing any other nodes as
routers. Each communication link can be active and transmit with a
constant power $P$ or remain silent. Throughput and rate-per-link
are the network parameters which are of concern to us. Despite the
randomness of the channel, we are only interested in events that
occur with high probability, i.e., with probability tending to one
as $n \to \infty$. This deterministic approach to random wireless
networks has been also deployed in
\cite{Ebrahimi07,Kulkarni04,Franceschetti07j}.

In a previous work by the authors \cite{Ebrahimi07}, the throughput
maximization of a single-hop wireless network in a Rayleigh fading
environment has been investigated without any rate constraints. It
is shown that the maximum throughput scales like $\log n$. Also, a
decentralized link activation strategy, called the threshold-based
link activation strategy (TBLAS), is proposed that achieves this
scaling law. The throughput maximization using TBLAS yields an
average rate per active link that approaches zero as $n \to \infty$.
The same phenomenon has been observed in
\cite{Gowaikar05,Gupta00,Toumpis04,Franceschetti07j}. Since most of
the existing efficient channel codes are designed for moderate
rates, it is a drawback for a system to have zero-approaching rates.
Thus, from a practical point of view, it is appealing to assign
constant rates to active communication links. In
\cite{Grossglauser02}, it is shown that a nondecreasing
rate-per-node is achievable when nodes are mobile.

In this paper, we consider the problem of rate-constrained
throughput maximization in a Rayleigh fading environment. More
specifically, the objective is to maximize the number of active
links such that each active link can transmit with a constant rate
$\lambda$. We derive an upper bound that shows the number of active
links scales at most like $\frac{1}{\lambda} \log n$. To obtain a
lower bound, first, we examine the simple TBLAS of \cite{Ebrahimi07}
and show that it is capable of guaranteeing rate-per-links equal to
$\lambda$. The number of active links provided by this method scales
like $\Theta(\log n)$. The scaling factor is close to the optimum
when $\lambda$ is small. However, as $\lambda$ grows large, the
scaling factor decays exponentially with $\lambda$, making it far
below the upper bound $\frac{1}{\lambda}$. This inspires developing
an improved link activation strategy that works well for large
values of desired rates, as well. To this end, we propose a
double-threshold-based link activation strategy (DTBLAS).

DTBLAS is attained by adding an interference management phase to
TBLAS. This is done by choosing from good enough links only those
with small enough mutual interference. The analysis of DTBLAS is
more complicated than that of TBLAS. However, it can be carried out
using some results from the random graph theory. It is shown that
DTBLAS performs very close to the optimum. Indeed, its performance
reaches the upper bound when the demanded rate approaches $\infty$
or $0$. This shows the asymptotic optimality of DTBLAS for the
rate-constrained throughput maximization problem.

In all scenarios described above, the interference power is much
larger than the noise power and the rates become independent of
\emph{signal-to-noise ratio} (\emph{SNR}). In other words, the
network performs in an interference-limited regime. A natural
question is whether it is possible to have rate-per-links which
depend on the \emph{SNR}. The importance of this scenario, which is
called the noise-limited regime, is that the transmission rate
$\lambda$ can be adjusted by adjusting the transmission power $P$.
We show that the answer to the above question is affirmative and the
noise-limited regime can be realized by using DTBLAS. However, the
throughput achieved by this method scales like $\frac{\log n}{\log
\log n}$, which is by a multiplicative factor of ${\log \log n}$
less than what is achievable in an interference-limited regime.

It is worth mentioning that link activation strategies studied in
this paper can be considered as special power allocation schemes.
The problem of throughput maximization via power allocation is a
challenging problem for which only suboptimum solutions have been
reported \cite{Julian02, Boche04,Jindal07c}. However, variations of
this problem have been extensively studied in the literature, where
the on-off scheme has frequently appeared. Recently, for a
decentralized utility-based network\footnote{Each node maximizes a
locally computed network average throughput conditioned on its own
channel gain.}, it is shown that the optimum power allocation
follows an on-off paradigm when the number of links is large
\cite{Jamshid07}. The on-off power allocation has been also used in
\cite{Gesbert07rawnet,Gesbert07} for a cellular network in which the
number of cell (links) are limited, but in each cell there are
infinite number of users to choose from. For cellular systems, a
distributed joint power allocation and scheduling has been proposed
in \cite{Kianiwcnc07}, in which again an on-off strategy is
followed.

The rest of the paper is organized as follows: In Section
\ref{system}, network model and problem formulation are presented.
An upper bound on the throughput is derived in
Section~\ref{upperbound}. In
Sections~\ref{LBDecentralized}~and~\ref{LBCentralized},
achievability results via decentralized and centralized schemes are
presented. Some optimality results are provided in
Section~\ref{optimality}. The operation of the network in a
noise-limited regime is investigated in Section \ref{noise-limited}.
Finally, the paper is concluded in Section \ref{conclusion}.

\emph{Notation:} $\mathcal{N}_n$ represents the set of natural
numbers less than or equal to $n$; $\log ( \cdot )$ is the natural
logarithm function; $\lfloor x \rfloor$ denotes the largest integer
less than or equal to $x$; $\chi^2(M)$ represents the chi-squared
distribution with $M$ degrees of freedom; $\textrm{P} (A)$ denotes
the probability of event $A$; $\textrm{E}(x)$ and $\textrm{Var}(x)$
represent the expected value and the variance of the random variable
$x$, respectively; $\approx$ means approximate equality; for any
functions $f(n)$ and $h(n)$, \mbox{$h(n)=O(f(n))$} is equivalent to
$ \lim_{n \to \infty} \left | {h(n)}/ {f(n)} \right |< \infty$,
\mbox{$h(n)=o(f(n))$} is equivalent to $ \lim_{n \to \infty} \left |
{h(n)}/{f(n)} \right | =0$, \mbox{$h(n)=\omega(f(n))$} is equivalent
to $ \lim_{n \to \infty} \left | {h(n)}/{f(n)} \right | = \infty$,
\mbox{$h(n)=\Theta(f(n))$} is equivalent to $ \lim_{n \to \infty}
\left | {h(n)}/{f(n)} \right | = c$, where $0<c<\infty$, and
\mbox{$h(n)\sim f(n)$} is equivalent to $ \lim_{n \to \infty}
{h(n)}/{f(n)}  = 1$; an event $A_n$ holds \emph{asymptotically
almost surely} (a.a.s) if $\textrm{P} (A_n) \to 1$ as $ n \to
\infty$.

\section{Network Model and Problem Formulation}\label{system}

The network model is the same as in \cite{Ebrahimi07}; however, we
repeat it here for completeness. We consider a wireless
communication network with $n$ pairs of transmitters and receivers.
These $n$ communication links are indexed by the elements of
$\mathcal{N}_n$. Each transmitter aims to send data to its
corresponding receiver in a single-hop fashion. The transmit power
of link $i$ is denoted by $p_i$. It is assumed that the links follow
an on-off paradigm, i.e., $p_i \in \{ 0,\,P \} $, where $P$ is a
constant. Hence, any power allocation scheme translates to a
\emph{link activation strategy (LAS)}. Any LAS yields a set of
active links $\mathcal{A}$, which describes the transmission powers
as
\begin{equation}
p_i=\left \{
\begin{tabular}{lcl}
$P$ & if & $i \in \mathcal{A} $\\
$0$ & if & $i \notin \mathcal{A}$
\end{tabular}.
\right .
\end{equation}

The channel between transmitter $j$ and receiver $i$ is
characterized by the coefficient $g_{ji}$. This means the received
power from transmitter $j$ at the receiver $i$ equals $g_{ji}p_j$.
We assume that the channel coefficients are \emph{independent
identically distributed} (i.i.d.) random variables drawn from an
exponential pdf, i.e., $f(x)=e^{-x}$, with mean $\mu=1$ and variance
$\sigma^2=1$. This channel model corresponds to a Rayleigh fading
environment. We refer to the coefficients $g_{ii}$ and $g_{ji}$ ($j
\neq i$) as \emph{direct channel coefficients} and \emph{cross
channel coefficients}, respectively.

We consider an additive white Gaussian noise (AWGN) with limited
variance $\eta$ at the receivers. The transmit \emph{SNR} of the
network is defined as
\begin{equation}
\rho=\dfrac{P}{\eta}.
\end{equation}
The receivers are conventional linear receivers, i.e., without
multiuser detection. Since the transmissions occur simultaneously
within the same environment, the signal from each transmitter acts
as interference for other links. Assuming Gaussian signal
transmission from all links, the distribution of the interference
will be Gaussian as well. Thus, according to the Shannon capacity
formula \cite{cover}, the maximum supportable rate of link $i \in
\mathcal{A}$ is obtained as
\begin{equation}\label{ratedef}
r_i (\mathcal{A})=\log \left( 1+ \gamma_i(\mathcal{A}) \right) \;
\textrm{nats/channel use},
\end{equation}
where
\begin{equation}\label{SINRdef}
\gamma_i(\mathcal{A})=\dfrac{g_{ii}}{1/\rho+\sum_{\substack{j \in \mathcal{A}\\
j \neq i}} g_{ji}}
\end{equation}
is the \emph{signal-to-interference-plus-noise ratio} (\emph{SINR})
of link $i$.

As a measure of performance, in this paper we consider the
throughput of the network, which is defined as
\begin{equation}\label{throughputdef}
T(\mathcal{A})=\sum_{i \in \mathcal{A}} r_i(\mathcal{A}).
\end{equation}
Also, the \emph{average rate per active link} is defined as
\begin{equation}\label{averagerate}
\bar{r}(\mathcal{A})=\dfrac{T(\mathcal{A})}{|\mathcal{A}|}.
\end{equation}
In this paper, wherever there is no ambiguity, we drop the
functionality of $\mathcal{A}$ from the network parameters and
simply refer to them as $r_i$, $\gamma_i$, $T$, or $\bar{r}$.

Throughout the paper, we assume all active links transmit with a
same constant rate $\lambda$. In this case, the throughput becomes
proportional to the number of active links, i.e.,
$T(\mathcal{A})=|\mathcal{A}|\lambda$. Hence, the problem of
throughput maximization becomes equivalent to maximizing the number
of active links subject to a constraint on the rate of active links,
i.e.,
\begin{equation}\label{constrainedproblem}
\begin{tabular}{cc}
$\displaystyle \max_{\mathcal{A} \subseteq \mathcal{N}_n}$ & $ |\mathcal{A}|$\\
$\textrm{s.t.}$ & $r_i(\mathcal{A}) \geq \lambda, \quad \forall i
\in \mathcal{A}$
\end{tabular}.
\end{equation}
This problem is referred to as the \emph{rate-constrained throughput
maximization}. We denote the throughput corresponding to the maximum
value of this problem by $T_c^*$.

Due to the nonconvex and integral nature of the throughput
maximization problem, its solution is computationally intensive.
However, in this paper we propose and analyze LASs which lead to
efficient solutions for the above problem. Indeed, we first show
that the decentralized method of \cite{Ebrahimi07} is a.a.s. optimum
when $\lambda$ is vanishingly small. Then, we propose a new LAS
which is asymptotically optimum for large values as well as small
values of $\lambda$. Also, for moderate values of $\lambda$, there
is a small gap between the performance of the proposed LAS and a
derived upper bound. This shows the closeness of its performance to
the optimum.

For simplicity of notation, we denote the number of active links by
$k$ instead of $ |\mathcal{A}|$. Motivated by the result of
\cite{Ebrahimi07} that shows the maximum throughput scales like
$\log n$, we introduce the following definitions. The scaling factor
of the throughput is defined as
\begin{equation}\label{throughputscalingfactor}
\tau=\lim_{n \to \infty} \dfrac{T}{\log n},
\end{equation}
Similarly, the scaling factor of the number of active links is defined as 
\begin{equation}\label{nousersscalingfactor}
\kappa=\lim_{n \to \infty} \dfrac{k}{\log n}.
\end{equation}

\section{Upper Bound}\label{upperbound}

In this section, we obtain an upper bound on the optimum solution of
(\ref{constrainedproblem}). This upper bound can be either presented
as an upper bound on the throughput or as an upper bound on the
number of active links.

\begin{theorem}\label{throughputupperbound}
Assume $\mathcal{A}_c^*$ is the solution to the rate-constrained
throughput maximization \eqref{constrainedproblem} and
$k_c^*=|\mathcal{A}_c^*|$. Then, the associated throughput and the
scaling factor of $k_c^*$ a.a.s. satisfy
\begin{eqnarray}
T_c^* & < & \log n - \log \log n + c,\\
\kappa_c^* & < & \dfrac{1}{\lambda},
\end{eqnarray}
for some constant $c$.
\end{theorem}

\begin{proof}
For a randomly selected set of active links $\mathcal{A}$ with
$|\mathcal{A}|=k$, the interference term $I_i= \sum_{\substack{j \in \mathcal{A}\\
j \neq i}} g_{ji}$ in the denominator of (\ref{SINRdef}) has
$\chi^2(2k-2)$ distribution. Hence, we have
\begin{eqnarray}
\nonumber \textrm{P}(\gamma_i > x) & = & \int_0^\infty \textrm{P}\left(\gamma_i > x | I_i=z \right) f_{I_i}(z) dz \\
\nonumber & = & \int_0^\infty e^{-x (1/\rho+z)} \dfrac{z^{k-2}e^{-z}}{(k-2)!}dz \\
& = & \dfrac{e^{-x/\rho}}{(1+x)^{k-1}}.\label{ccdfSINR}
\end{eqnarray}

Assume $\mathcal{L}_1$ is the event that there exists at least one
set $\mathcal{A} \subseteq \mathcal{N}_n$ with $|\mathcal{A}|=k$
such that the constraints in (\ref{constrainedproblem}) are
satisfied. Also, assume $\gamma_0$ is a quantity that satisfies
$\lambda= \log ( 1 + \gamma_0 )$. We have
\begin{eqnarray}
\textrm{P} (\mathcal{L}_1) & \stackrel{(a)}{\leq} & \binom{n}{k}
\left( \textrm{P} (r_i \geq
\lambda) \right ) ^k \\
& = & \binom{n}{k} \left( \textrm{P} (\gamma_i \geq \gamma_0)
\right ) ^k\\
& \stackrel{(\ref{ccdfSINR})}{=} & \binom{n}{k} \dfrac{e^{-
\gamma_0k/\rho}}{(1+\gamma_0)^{k(k-1)}}\\
& \stackrel{(b)}{\leq} & \left( \dfrac{ne}{k} \right) ^ k
\dfrac{e^{-
\gamma_0k/\rho}}{(1+\gamma_0)^{k(k-1)}}\\
& = & e ^ {k(\log n - \log k - \lambda k + \lambda +1 -
\gamma_0/\rho )},
\end{eqnarray}
where (a) is due to the union bound and (b) is the result of
applying the Stirling's approximation for the factorial. It can be
verified that there exists a constant $c$ such that if $k\lambda =
\log n - \log \log n + c $, then, the above upper bound approaches
zero for $n \to \infty$. Hence, for the event $\mathcal{L}_1$ to
have non-zero probability, we should a.a.s. have
\begin{equation}
k\lambda < \log n - \log \log n + c.
\end{equation}
This inequality holds for any feasible number of active links. By
choosing $k=k_c^*$, the upper bounds in the lemma are proved.
\end{proof}

\section{Lower Bound: A Decentralized Approach}\label{LBDecentralized}

To derive a lower bound, in this section, we consider the
threshold-based link activation strategy (TBLAS) originally proposed
in \cite{Ebrahimi07}.

\emph{TBLAS:} For a threshold $\Delta$, choose the set of active
links according to the following rule
\begin{equation}\label{strategy1}
i \in \mathcal{A} \quad \textrm{iff} \quad g_{ii} > \Delta.
\end{equation}

As it is seen, in TBLAS each link only needs to know its own direct
channel gain. If a direct channel gain is above the threshold
$\Delta$, the corresponding link is active; otherwise, it remains
silent. The value of $\Delta$ determines the achievable throughput.
We show that by proper choose of the threshold, TBLAS provides a
solution for the rate-constrained throughput maximization. The
importance of TBLAS is that it can be implemented in a decentralized
fashion.

Let us denote the achieved throughput of TBLAS by $T_{_{TBLAS}}$.
The following results are proven for TBLAS in \cite{Ebrahimi07}:
\begin{eqnarray}
\label{appsumraterayleigh} {T_{_{TBLAS}}} & \sim & n e^{-\Delta}
\log
\left( 1+ \dfrac{\Delta}{n e^{-\Delta}} \right ),\\
\label{activelinksrayleigh} k_{_{TBLAS}} & \sim & n e^{-\Delta},\\
\label{kbound} | k_{_{TBLAS}} - ne^{-\Delta} | & < & \xi \sqrt
{ne^{-\Delta}},\quad a.a.s.
\end{eqnarray}
where the last inequality holds for any $\xi=\omega(1)$.

A necessary condition for the rate of active links being equal to
$\lambda$ is $\bar{r}_{_{TBLAS}}=\lambda$, where
$\bar{r}_{_{TBLAS}}$ is the average rate per active link achieved by
TBLAS. Hence, we should choose $\Delta$ such that the throughput and
the number of active links both become proportional to $\log n$. The
following lemma shows how to realize such a scenario.
\begin{lemma}\label{strategy1constrained}
Assume the activation threshold for TBLAS is chosen to be $\Delta =
\log n - \log \log n - \log \alpha$ for some $\alpha > 0$. Then,
a.a.s. we have
\begin{eqnarray}
\tau_{_{TBLAS}} & = & \alpha \log \left(1 + \frac{1}{\alpha}\right)\\
\kappa_{_{TBLAS}} & = & \alpha \label{kscaling1}\\
\bar{r}_{_{TBLAS}} & = & \log \left(1 + \frac{1}{\alpha}\right) +
o(1).\label{rateperlink1}
\end{eqnarray}
\end{lemma}

\begin{proof}
With the specified value of $\Delta$, we have $n e^{-\Delta}=\alpha
\log n$.  The values of $\tau_{_{TBLAS}}$ and $\kappa_{_{TBLAS}}$
are readily obtained by substituting this value in
(\ref{appsumraterayleigh}) and (\ref{activelinksrayleigh}) and using
the definitions \eqref{throughputscalingfactor} and
\eqref{nousersscalingfactor}, respectively. The value of
$\bar{r}_{_{TBLAS}}$ is obtained by using the definition
\eqref{averagerate}.
\end{proof}

Lemma \ref{strategy1constrained} indicates that by a proper choose
of $\alpha$, an average rate per active link equal to $\lambda$ is
achievable; however, it does not guarantee that all active links can
support this rate. In other words, one may ask whether TBLAS is
capable of satisfying the constraints in problem
(\ref{constrainedproblem}). The following lemma addresses this issue
and shows that a.a.s. the rate of all active links are highly
concentrated around the average rate per active link.

\begin{lemma}\label{concentration1}
Assume the activation threshold for TBLAS is chosen to be $\Delta =
\log n - \log \log n - \log \alpha$ for some $\alpha > 0$. Then,
a.a.s. we have
\begin{equation}
|r_i - \bar{r}| <  2\sqrt{\dfrac{\log \log n}{\alpha^3 \log n}}(1+o
( 1 )), \quad \forall i \in \mathcal{A},
\end{equation}
where $\bar{r} = \log \left( 1 + \dfrac{1}{\alpha} \right )$.
\end{lemma}

To prove the lemma, we need the following result about the
\emph{central limit theorem} (\emph{CLT}) for large deviations.

\begin{theorem}[\cite{LimitTheorems}]\label{clt}
Let $\left\{ Y_m \right\}$ be a sequence of i.i.d. random variables.
Suppose that $Y_1$ has zero mean and finite positive variance $\nu$
and satisfies Cram\'er's condition\footnote{A random variable $Y$
satisfies the Cram\'er's condition if its moment-generating function
exists in some interval with the center at the origin.}. For
$Z_m=\frac{1}{\sqrt{ m \nu}}\sum_{j=1}^m Y_j$, define
$F_m(y)=\textrm{P}(Z_m < y)$. If $y \geq 0$, $y=O(m^{1/6})$, then
\begin{equation}\label{clteq}
1-F_m(y)=[1-\Phi(y)]\exp\left( \frac{\theta_3 y^3}{6  \sqrt{m \nu^3
}} \right) + O \left(  \frac{e^{-y^2/2}}{\sqrt{m}} \right),
\end{equation}
where $\Phi(y)$ is the cdf of normal distribution and $\theta_3
=\textrm{E}(Y_1^3)$.
\end{theorem}

\begin{IEEEproof}[Proof of Lemma \ref{concentration1}]
From the definition of $r_i$ and the concavity of the $\log ( \cdot
)$ function, we have
\begin{eqnarray}
|r_i - \bar{r}| & = & \left|\log (1+ \gamma_i) -
\log\left(1+\frac{1}{\alpha}\right)\right|\\
& \leq & \left|\gamma_i - \frac{1}{\alpha}\right|.
\end{eqnarray}
Thus, to prove the lemma, it is sufficient to prove that a.a.s.
\begin{equation}
\left|\gamma_i - \frac{1}{\alpha}\right| <  2\sqrt{\dfrac{\log \log
n}{\alpha^3 \log n}}(1+o ( 1 )), \quad \forall i \in \mathcal{A},
\end{equation}
or equivalently
\begin{equation}\label{equivalentbound}
x_- < \gamma_i < x_+,
\end{equation}
where
\begin{equation}\label{cdfparameter}
x_{\pm} = \dfrac{1}{\alpha}\left( 1 \pm  2\sqrt{\dfrac{\log \log
n}{\alpha \log n}}(1+o ( 1 )) \right).
\end{equation}
Here, we just prove the left-side inequality in
\eqref{equivalentbound}. The other side can be proved in a similar
manner.

Let $\mathcal{L}_2$ denote the event that
\begin{equation}\label{SINRlowerbound}
\gamma_i > x_-,  \quad \forall  i \in \mathcal{A}.
\end{equation}
In the following, we show that $\textrm{P}(\mathcal{L}_2) \to 1$ as
$n \to \infty$.

Denoting the cdf of $\gamma_i$ conditioned on $|\mathcal{A}|=k$ by
$F_{\gamma}(x,k)$, the probability of the event $\mathcal{L}_2$ is
obtained as
\begin{eqnarray}
\textrm{P}(\mathcal{L}_2)& = & \sum_{k=0}^n
\textrm{P}(|\mathcal{A}|=k)\textrm{P}(\mathcal{L}_2||\mathcal{A}|=k)\\
& \stackrel{(a)}{=} & \sum_{k=0}^n \textrm{P}(|\mathcal{A}|=k)
\left(1-F_{\gamma}(x_-,k)\right)^k\\
& \stackrel{(b)}{\geq} & \sum_{k=k_-}^{k_+}
\textrm{P}(|\mathcal{A}|=k)
\left(1-F_{\gamma}(x_-,k)\right)^k\\
& \stackrel{(c)}{>} & \left(1-F_{\gamma}(x_-,k_+)\right)^{k_+}
\sum_{k=k_-}^{k_+} \textrm{P}(|\mathcal{A}|=k)
\\
& = & \left(1-F_{\gamma}(x_-,k_+)\right)^{k_+} \textrm{P}(k_- \leq
|\mathcal{A}| \leq k_+),\label{khastehshodam}
\end{eqnarray}
where (a) is because the channel gains are independent, (b) is valid
for any $0 \leq k_- \leq k_+ \leq n$ and (c) is due to the fact that
$\left(1-F_{\gamma}(x,k)\right)^k$ is a decreasing function of $k$.
According to (\ref{kbound}), by choosing
\begin{eqnarray}
k_\pm & =  & ne^{-\Delta} \pm \xi \sqrt{ne^{-\Delta}} \\
& = & \alpha \log n \pm \xi \sqrt{\alpha \log n}, \label{kpm}
\end{eqnarray}
for some $\xi \to \infty$, we have $ \textrm{P}(k_- \leq
|\mathcal{A}| \leq k_+) \to 1 $. Hence, to prove $
\textrm{P}(\mathcal{L}_2) \to 1$, it is enough to show that $
\left(1-F_{\gamma}(x_-,k_+)\right)^{k_+} \to 1$. However, due to the
inequality
\begin{equation}
\left(1-F_{\gamma}(x_-,k_+)\right)^{k_+} \geq 1-k_+
F_{\gamma}(x_-,k_+),
\end{equation}
it is enough to show that
\begin{equation}\label{prob1cond}
k_+ F_{\gamma}(x_-,k_+) \to 0.
\end{equation}
To prove (\ref{prob1cond}), we provide an upper bound on $ k_+
F_{\gamma}(x_-,k_+) $ and show that it approaches zero as $n \to
\infty$. We have
\begin{eqnarray}
\nonumber F_{\gamma}(x_-,k_+)& = & \textrm{P}(\gamma_i \leq x_- | |\mathcal{A}|=k_+ )\\
\nonumber & \stackrel{(a)}{=} & \textrm{P} \left( \dfrac{g_{ii}}{1/\rho+\sum_{\substack{j =1\\
j \neq i}}^{k_+} g_{ji}} \leq x_- \right)\\
\nonumber & = & \textrm{P} \left(\sum_{\substack{j =1\\
j \neq i}}^{k_+} g_{ji} \geq \dfrac{g_{ii}}{x_-} -\frac{1}{\rho} \right)\\
& \stackrel{(b)}{<} & \textrm{P} \left(\sum_{\substack{j =1\\
j \neq i}}^{k_+} g_{ji} \geq \dfrac{\Delta}{x_-} -\frac{1}{\rho}
\right),\label{upperFGamma}
\end{eqnarray}
where (a) is based on $\mathcal{A}=\left\{ 1,\, \cdots,\, k_+
\right\}$, which has been assumed for simplicity of notation, and
(b) is due to the fact that, in TBLAS, $g_{ii}
> \Delta$ for any $i \in \mathcal{A}$. Let us define $Y_j=g_{ji}-1$,
which has the variance $\nu=1$. Thus, the right-hand-side (RHS) of
(\ref{upperFGamma}) translates to the complementary cdf of
$Z=\dfrac{1}{\sqrt{k_+ - 1}} \sum_{\substack{j =1\\
\nonumber j \neq i}}^{k_+} Y_j$, i.e. (\ref{upperFGamma}) can be
rewritten as
\begin{eqnarray}\label{rhs1}
F_{\gamma}(x_-,k_+)& < &  1 - \textrm{P} (Z \leq y),
\end{eqnarray}
where
\begin{eqnarray}\label{defy1}
y & = & \dfrac{\frac{\Delta}{x_-} -\frac{1}{\rho} -
(k_+-1)}{\sqrt{k_+-1}}.
\end{eqnarray}
By substituting $\Delta = \log n - \log \log n - \log \alpha$ and
the value of $x_-$ from \eqref{cdfparameter} into \eqref{defy1}, we
obtain
\begin{equation}\label{valuey1}
y = 2 \sqrt{\log \log n} (1+ o(1)).
\end{equation}

Since $Y_j$ is a shifted exponential random variable, its
moment-generating function exists around zero and the Cram\'er's
condition is satisfied. Also, by choosing $m=k_+-1$ we have
$y=O(m^{1/6})$. Hence, the result of Theorem \ref{clt} can be
applied to calculate the complementary cdf of $Z$. Consequently, by
using \eqref{clteq} with $\theta_3=\textrm{E}(Y_j^3)=2$,
\eqref{rhs1} can be rewritten as
\begin{equation}\label{rhs2}
F_{\gamma}(x_-,k_+) <
[1-\Phi(y)]\exp\left(\dfrac{y^3}{3\sqrt{k_+-1}}\right) +O\left(
\dfrac{e^{-y^2/2}}{\sqrt{k_+-1}} \right).
\end{equation}
Noting that $y^3=o(\sqrt{k_+})$ and using the inequality $1-\Phi(y)
< \frac{e^{-y^2/2}}{y}$, from \eqref{rhs2} and \eqref{valuey1}, we
conclude that
\begin{eqnarray}
k_+ F_{\gamma}(x_-,k_+) & < &  k_+ \dfrac{e^{-y^2/2}}{y}\\
& = & \exp\left( - \log \log n (1+o(1)) \right).
\end{eqnarray}
It is clear that the above upper bound approaches zero as $n \to
\infty$. Hence, $\textrm{P}(\mathcal{L}_2) \to 1$ and the proof is
complete.
\end{IEEEproof}

Lemma \ref{concentration1} shows that with the specified threshold
for TBLAS, all active links can transmit with rate $\lambda=\log ( 1
+ \frac{1}{\alpha} )$. Hence, TBLAS provides a solution, albeit
suboptimum, for the problem~(\ref{constrainedproblem}). Lemmas
\ref{strategy1constrained} and \ref{concentration1} reveal the
following relation between the demanded rate $\lambda$ and
$\kappa_{_{TBLAS}}$ as well as $\tau_{_{TBLAS}}$
\begin{eqnarray}
\kappa_{_{TBLAS}}=\dfrac{1}{e^{\lambda}-1},\label{lktradeoff}\\
\tau_{_{TBLAS}}=\dfrac{\lambda}{e^{\lambda}-1}.\label{tlstrategy1}
\end{eqnarray}
Noting that for small values of $\lambda$, the RHS of
(\ref{lktradeoff}) can be approximated as $\frac{1}{\lambda}$ and
using the upper bound in Theorem~\ref{throughputupperbound}, it
turns out that TBLAS is close to the optimum for small values of
$\lambda$.

\section{Lower Bound: A Centralized Approach}\label{LBCentralized}

Although TBLAS enjoys the simplicity of decentralized
implementation, its performance is far from the optimum. This can be
seen by comparing the upper bound in
Theorem~\ref{throughputupperbound} and the achievability result in
(\ref{lktradeoff}). A reason for this suboptimality is that the
mutual interference of the active links is not considered in
choosing $\mathcal{A}$. In this section, we provide an LAS that
performs close to the upper bound in Theorem
\ref{throughputupperbound} and turns out to be asymptotically
optimum when $\lambda$ is very large or very small. We name this
method double-threshold-based LAS (DTBLAS).

\emph{DTBLAS:} For the thresholds $\Delta$ and $\delta$
\begin{enumerate}
\item Choose the largest set $\mathcal{A}_1 \subseteq \mathcal{N}_n$ such that  $g_{ii} > \Delta$ for all $i \in \mathcal{A}_1$.
\item Choose the largest set $\mathcal{A}_2 \subseteq \mathcal{A}_1 $ such that $g_{ij} \leq \delta$ and $g_{ji} \leq
\delta$ for all $ i,\,j \in \mathcal{A}_2$.
\end{enumerate}
The set of active links is $\mathcal{A}=\mathcal{A}_2$.

This strategy chooses the links to be active in a two-phase
selection process; in the first phase, which is basically similar to
TBLAS, a subset $\mathcal{A}_1$ of the links with good enough direct
channel coefficients is chosen. In the second phase, which is the
interference management phase, a subset of links in $\mathcal{A}_1$
is chosen such that their mutual interferences are small enough.
Note that the second phase of the strategy requires full knowledge
of the channel coefficients. Hence, this scheme should be
implemented in a centralized fashion.

We aim to find $\Delta$ and $\delta$ such that the throughput is
maximized subject to the rate constraints of the active links.

For simplicity, we use the notation $k_i=|\mathcal{A}_i|$ for
$i=1,\,2$. Without loss of generality, assume
$\mathcal{A}_i=\{1,\,\cdots,\,k_i\}$. By using (\ref{ratedef}),
(\ref{SINRdef}), and (\ref{throughputdef}), and applying the
Jensen's inequality, the throughput is lower bounded as
\begin{equation}\label{lbound4}
T   \geq k_2  \log \left( 1+\dfrac{\Delta}{ 1/\rho+\frac{1}{k_2}
\sum_{i=1}^{k_2} I_i} \right ),
\end{equation}
where $I_i= \sum_{\substack{j=1\\ j \neq i}}^{k_2} g_{ji}$. Since
$g_{ji} \leq \delta$, the mean and variance of $I_i$ depend on
$\delta$. More precisely, we have
\begin{eqnarray}
\textrm{E} (I_i) & = & (k_2 -1) \hat{\mu},\label{meaninterference}\\
\textrm{Var} (I_i) & = & (k_2 -1)
\hat{\sigma}^2,\label{meanvariance}
\end{eqnarray}
where
\begin{eqnarray}
\hat{\mu} & = & \textrm{E} \left\{ g_{ji} | g_{ji} \leq \delta
\right\}=1-\dfrac{\delta e^{-\delta}}{1-e^{-\delta}},\label{conditionalmean}\\
\hat{\sigma}^2 & = & \textrm{Var} \left\{ g_{ji} | g_{ji} \leq
\delta \right\}=1-\dfrac{\delta^2
e^{-\delta}}{(1-e^{-\delta})^2}.\label{conditionalvariance}
\end{eqnarray}

Assume $\delta$ is a constant and $k_2 \to \infty$ as $n \to
\infty$. To simplify the RHS of (\ref{lbound4}), we apply the
Chebyshev inequality to obtain the upper bound
\begin{equation}\label{iupperbound}
\frac{1}{k_{2}} \sum_{i=1}^{k_{2}} I_i <(k_{2}-1)\hat{\mu} + \psi,
\end{equation}
which holds a.a.s. for any $\psi=\omega(1)$. Consequently, the lower
bound (\ref{lbound4}) becomes
\begin{equation}\label{lbound8}
T   \geq k_2  \log \left( 1+\dfrac{\Delta}{\hat{\mu} k_2 + \psi}
\right ) \;\qquad a.a.s.
\end{equation}
Note that the constant $1/\rho - \hat{\mu}$ is absorbed in the
function $\psi$. Since  $\psi$ can be chosen arbitrarily small, say
with an order smaller than $\hat{\mu} k_2$, we can rewrite
\eqref{lbound8} as
\begin{equation}
T \geq T_{_{DTBLAS}},
\end{equation}
where
\begin{equation}\label{lbound9}
T_{_{DTBLAS}}   = k_2 \left( \log \left( 1+\dfrac{\Delta}{\hat{\mu}
k_2} \right ) + o(1) \right) \;\qquad a.a.s.
\end{equation}
denotes the throughput achievable by DTBLAS.

 Since $k_2$ is a random variable, the right hand side of
(\ref{lbound9}) is a random variable as well. However, the following
discussion shows that $k_2$ is highly concentrated around a certain
value. Hence, it can be treated as a deterministic value.

Construct an undirected graph $G(\mathcal{A}_1,\,\boldsymbol{E})$
with vertex set $\mathcal{A}_1$ and the adjacency matrix $
\boldsymbol{E}=[e_{ij}] $ defined as
\begin{equation}\nonumber
e_{ij}=\left \{
\begin{tabular}{lcl}
1 & ; & $g_{ij} \leq \delta$ and $g_{ji} \leq \delta$\\
0 & ; & otherwise
\end{tabular}.
\right .
\end{equation}
The probability of having an edge between vertices $i$ and $j$, when
$g_{ji}$ and $g_{ij}$ have exponential distribution, equals
\begin{equation}\label{rgraphprob}
p=\left( 1- e^{-\delta} \right )^2.
\end{equation}\sloppy
The definition of $G$ implies that $G \in \mathcal{G}(k_1,p) $, where $
\mathcal{G}(k_1,p)$, which is a well-studied object in the literature \cite{rgas}, is the
family of $k_1$-vertex \emph{random graphs} with edge probability $p$.

In the second phase of DTBLAS, we are interested to choose the
maximum number of links whose cross channel coefficients are smaller
than $\delta$. This is equivalent to choosing the largest complete
subgraph\footnote{A complete graph is a graph in which every pair of
vertices are connected by an edge.} of $G$. The size of the largest
complete subgraph of $G$ is called its clique number and denoted by
$\textrm{cl}(G)$. The above discussion yields
\begin{equation}\label{k2concentrate}
k_2=\textrm{cl}(G), \quad \textrm{for some} \quad G \in
\mathcal{G}(k_1,p).
\end{equation}
Although the clique number of a random graph $G$ is a random
variable, the following result from random graph theory states that
it is concentrated in a certain interval.

\begin{theorem}\label{cliquenumber}
Let $ 0 < p < 1$ and $\epsilon>0$ be fixed. The clique number
$\textrm{cl}(G)$ of $G \in \mathcal{G}(m,p) $, for large values of
$m$, a.a.s. satisfies $s_1 \leq \textrm{cl}(G) \leq s_2$ where
\begin{equation}
s_i=\lfloor  2 \log _b m - 2 \log _b \log _b m(1-p) + 2 \log _b
(e/2) +1 + (-1)^i \epsilon/p  \rfloor, \quad i=1,\,2,
\end{equation}
$b=1/p$.
\end{theorem}\vspace{6pt}

\begin{proof}
The theorem is a direct result of Theorem 7.1 in
\cite{randomgraphs2}, which states a similar result for the
stability number of random graphs. Using the fact that the stability
number of a random graph $\mathcal{G}(m,\,p)$ is the same as the
clique number of a random graph $\mathcal{G}(m,\,1-p)$, the theorem
is proved.
\end{proof}

\begin{corol}\label{concentrationactive2}
Consider DTBLAS with parameters $\Delta$ and $\delta$. The number of
active links, $k_{_{DTBLAS}}=k_2$, a.a.s. satisfies $k_-'\leq
k_{_{DTBLAS}} \leq k_+'$, where
\begin{equation}\label{klimits2}
k_\pm '= \lfloor 2 \log_b ne^{-\Delta} - 2 \log_b \log_b
ne^{-\Delta}(1-\frac{1}{b}) + 2 \log _b (e/2) +1 \pm \epsilon/p +
o(1) \rfloor
\end{equation}
and $b=(1-e^{-\delta})^{-2}$.
\end{corol}

\begin{proof}
According to \eqref{kbound}, a.a.s. we have $k_1=ne^{-\Delta} + O
(\xi \sqrt{ne^{-\Delta}})$. Assuming $\xi =o(\sqrt{ne^{-\Delta}})$,
and by substituting this value of $k_1$ into \eqref{k2concentrate}
and using Theorem~\ref{cliquenumber}, the corollary is proved.
\end{proof}

The next lemma indicates how to choose the thresholds $\Delta$ and
$\delta$ such that the throughput and the number of active links
both become proportional to $\log n$. As a result, a constant
average rate per active link is achieved.

\begin{lemma}\label{constantparameters}
Assume the threshold $\Delta$ for DTBLAS is chosen to be
\begin{equation}\label{deltavalue2}
\Delta = (1-\alpha') \log n (1 + o(1)),
\end{equation}
for some $\alpha'> 0$ and $\delta$ is a constant. Then, a.a.s. we
have
\begin{eqnarray}
\kappa_{_{DTBLAS}} & = & \dfrac{- \alpha'}{\log \left (
1-e^{-\delta}
\right)},\label{noactivelinksc}\\
\tau_{_{DTBLAS}} & = &  \dfrac{- \alpha'}{\log \left ( 1-e^{-\delta}
\right)} \log \left( 1-\dfrac{(1-\alpha') \log \left ( 1-e^{-\delta}
\right)}{\alpha' \left( 1- \dfrac{\delta
e^{-\delta}}{1-e^{-\delta}} \right)} \right), \label{throughputsc}\\
\bar{r}_{_{DTBLAS}} & = & \log \left( 1-
\dfrac{(1-\alpha')\log(1-e^{-\delta})}{\alpha' \hat{\mu}} \right ) +
o(1).
\end{eqnarray}
\end{lemma}

\begin{proof}
For the number of active links, we have
\begin{eqnarray}
k_{_{DTBLAS}} & \stackrel{(a)}{\sim} & 2 \log_b ne^{-\Delta}\\
& \stackrel{(b)}{=} & \dfrac{- \alpha' (1+o(1))}{\log \left (
1-e^{-\delta} \right)}  \log n, \label{activelinksstrategy2}
\end{eqnarray}
where (a) is based on Corollary~\ref{concentrationactive2} and (b)
is obtained by using \eqref{deltavalue2}. From
\eqref{activelinksstrategy2}, and by using the definition
\eqref{nousersscalingfactor}, $\kappa_{_{DTBLAS}}$ is obtained as
given in \eqref{noactivelinksc}.

The number of active links in \eqref{activelinksstrategy2} can be
used along with the value of $\Delta $ in \eqref{deltavalue2} to
rewrite (\ref{lbound9}) as
\begin{equation}\label{achievablesumratecent}
T_{_{DTBLAS}} =\left[ \dfrac{- \alpha'}{\log \left ( 1-e^{-\delta}
\right)} \log \left( 1-\dfrac{(1-\alpha') \log \left ( 1-e^{-\delta}
\right)}{\alpha' \hat{\mu}} \right) + o(1) \right] \log n.
\end{equation}
The scaling factor $\tau_{_{DTBLAS}}$, as given in the Lemma, is
obtained by using the value of $\hat{\mu}$ from
\eqref{conditionalmean} and applying the definition
\eqref{throughputscalingfactor}. The value of $\bar{r}_{_{DTBLAS}}$
is obtained by using the definition \eqref{averagerate}. This
completes the proof.
\end{proof}

According to this lemma, by proper choose of the constants $\alpha'$
and $\delta$, the average rate per active link $\bar{r}_{_{DTBLAS}}$
can be adjusted to be equal to the required rate $\lambda$. A
natural question is whether, under the specified conditions in
DTBLAS, all active links can support the rate $\lambda$. The
following lemma addresses this issue and shows that a.a.s. the rate
of all active links are highly concentrated around the average
value~$\bar{r}_{_{DTBLAS}}$.

\begin{lemma}\label{concentration2}
Consider DTBLAS with thresholds $\delta$ and $\Delta=(1-\alpha')
\log n$ for some $\alpha' > 0$. Then, a.a.s. we have
\begin{equation}
|r_i - \bar{r}| < c \sqrt{\dfrac{\log \log n}{ \log n}}(1+o ( 1 )),
\quad \forall i \in \mathcal{A},
\end{equation}
for some constant $c>0$, where $$\bar{r}=\log \left( 1 -
\dfrac{(1-\alpha')\log(1-e^{-\delta})}{ \alpha' \hat{\mu}} \right
).$$
\end{lemma}

\begin{proof}
See Appendix \ref{proofconcentration2}.
\end{proof}

According to Lemmas \ref{constantparameters} and
\ref{concentration2}, when maximizing the throughput of DTBLAS,
$\delta$ should be a constant and $\Delta$ is obtained from another
constant $\alpha'$. Hence, the rate-constrained throughput
maximization simplifies to an optimization problem with constant
parameters $\alpha'$ and~$\delta$. Assume $\gamma_0$ is a quantity
that satisfies $\lambda=\log(1+\gamma_0)$, i.e., $\gamma_0$ is the
required \emph{SINR} by the active links. Instead of the number of
active links, we can maximize the scaling factor of the number of
active links given in Lemma \ref{constantparameters}. Hence, the
rate-constrained throughput maximization problem
\eqref{constrainedproblem} is converted for DTBLAS to the following
optimization problem
\begin{eqnarray}
\max_{\alpha',\,\delta} & & \dfrac{- \alpha'}{\log \left (
1-e^{-\delta}
\right)} \label{constrainedopt}\\
\textrm{s.t.} & & -\dfrac{(1-\alpha') \log \left ( 1-e^{-\delta}
\right)}{\alpha' \left( 1- \dfrac{\delta e^{-\delta}}{1-e^{-\delta}}
\right)}=\gamma_0.\label{constraint2}
\end{eqnarray}
Note that in contrast to problem \eqref{constrainedproblem}, there
is only one constraint in this problem. However, according to
Lemma~\ref{concentration2}, this single constraint guarantees the
required rate for all active links. From the equality constraint
\eqref{constraint2}, parameter $\alpha'$ can be found in terms of
$\delta$ as
\begin{equation}\label{equivalentequalityconst}
\alpha'=\dfrac{- \log \left ( 1-e^{-\delta} \right)}{\gamma_0 \left(
1- \dfrac{\delta e^{-\delta}}{1-e^{-\delta}} \right) - \log \left (
1-e^{-\delta} \right)}.
\end{equation}
By substituting this value in the objective function
\eqref{constrainedopt}, we obtain the following equivalent
unconstrained optimization problem
\begin{equation}\label{unconstrainedopt}
\min_{\delta} \gamma_0 \left( 1- \dfrac{\delta
e^{-\delta}}{1-e^{-\delta}} \right) - \log \left ( 1-e^{-\delta}
\right).
\end{equation}
Consequently, $(\alpha'^*,\, \delta^*)$, the solution of
(\ref{constrainedopt}), can be obtained by first finding $\delta^*$
from (\ref{unconstrainedopt}) and then substituting it into
(\ref{equivalentequalityconst}) to obtain $\alpha'^*$. Due to the
complicated form of (\ref{unconstrainedopt}), it is not possible to
find $\delta^*$ analytically and it should be found numerically.

Fig. \ref{optimumdeltafigure} shows $\delta^*$ and $\alpha'^*$
versus $\lambda$.
\begin{figure}
\centering
\includegraphics[scale=0.85]{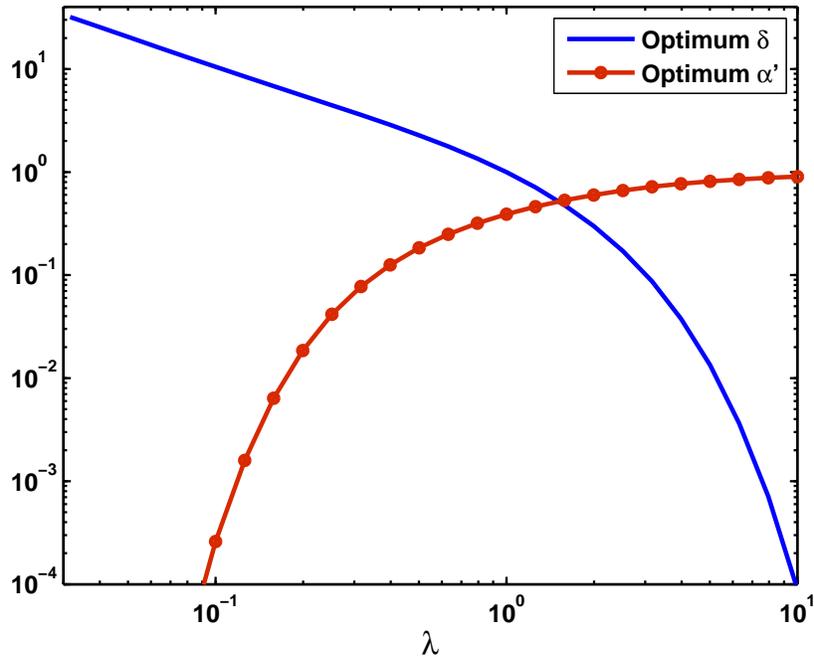}\\
\caption{Optimum of the threshold $\delta$ and the parameter
$\alpha'$ vs. the demanded rate
$\lambda$.}\label{optimumdeltafigure}
\end{figure}
The values of $\delta^*$ and $\alpha'^*$ can be replaced in
(\ref{throughputsc}) and (\ref{noactivelinksc}) to obtain the
maximum throughput scaling factor ($\tau^*_{_{DTBLAS}}$) as well as
the maximum scaling factor for the number of active links
($\kappa^*_{_{DTBLAS}}$). The value $\tau^*_{_{DTBLAS}}$ is shown in
Fig. \ref{throughputscfigure}.
\begin{figure}
\centering
\includegraphics[scale=0.85]{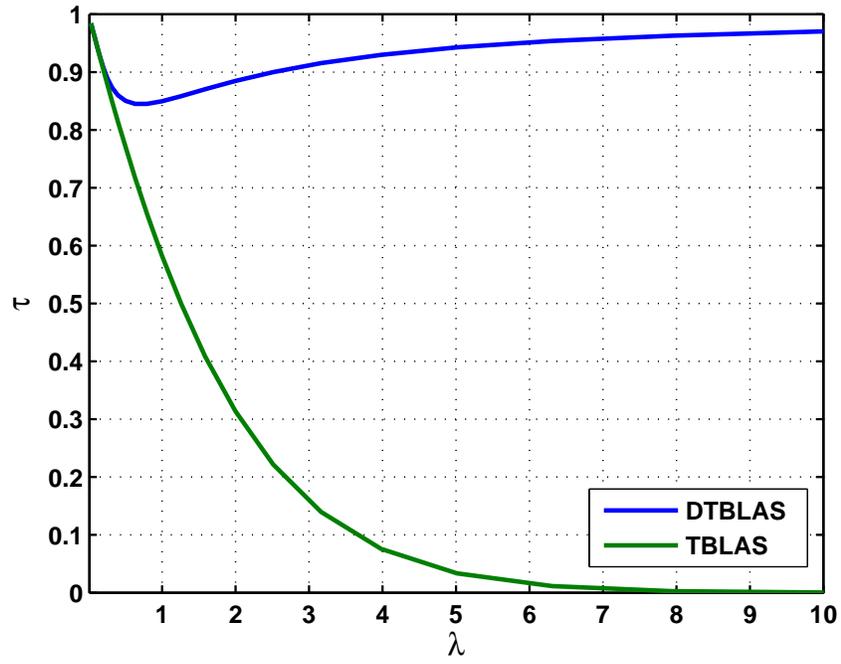}\\
\caption{Maximum throughput scaling factor vs. the demanded rate
$\lambda$.}\label{throughputscfigure}
\end{figure}
\begin{figure}
\centering
\includegraphics[scale=0.85]{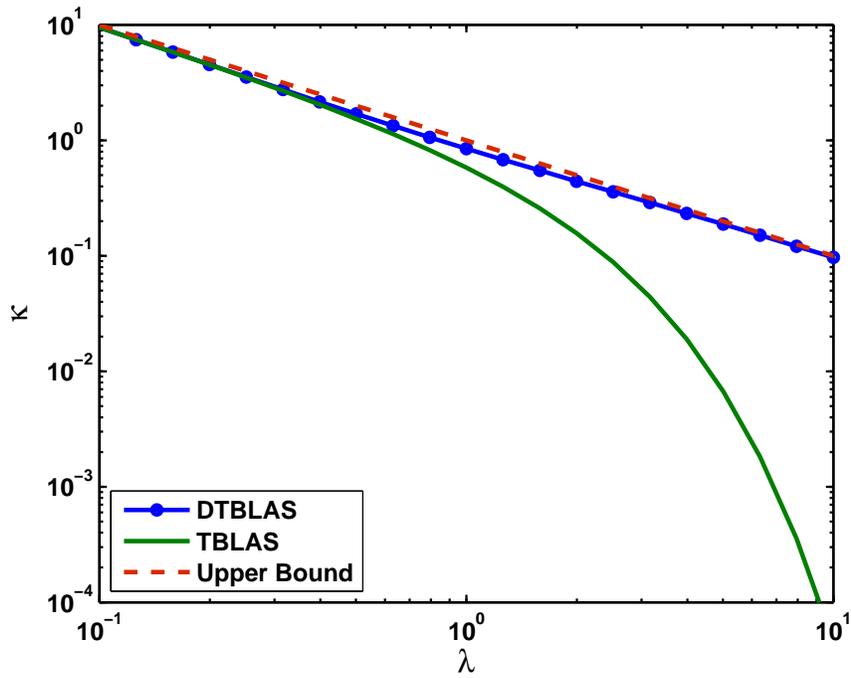}\\
\caption{Tradeoff between rate-per-link and the number of active
links.}\label{tradeofffigure}
\end{figure}
Depicted in the figure is also the throughput scaling factor of
TBLAS obtained from (\ref{tlstrategy1}). As it is observed, for
small values of $\lambda$, the performance of TBLAS and DTBLAS are
almost the same. However, as $\lambda$ grows larger, the scaling
factor of TBLAS approaches zero, but the scaling factor of DTBLAS
approaches 1. This shows some kind of optimality for DTBLAS which
will be later proven formally. Figure \ref{tradeofffigure}
demonstrates the tradeoff between the number of supported links and
the demanded rate-per-link for TBLAS and DTBLAS. The tradeoff curve
for TBLAS is obtained from (\ref{lktradeoff}). The upper bound from
Theorem \ref{throughputupperbound} is also plotted for comparison.
As observed, for a ceratin value of $\lambda$, DTBLAS can support
larger number of users, especially for larger values of $\lambda$.
Indeed, the tradeoff curve of DTBLAS is very close to the upper
bound. Specifically, for large values of $\lambda$, these two curves
coincide. This will be later proven formally.

\section{Optimality Results}\label{optimality}

Although the behaviour of DTBLAS is numerically described in Figs.
\ref{optimumdeltafigure}, \ref{throughputscfigure}, and
\ref{tradeofffigure}, it is possible and also insightful to obtain
closed form expressions for $\delta^*$ and $\alpha'^*$ as well as
$\kappa^*_{_{DTBLAS}}$ and $\tau^*_{_{DTBLAS}}$ when $\lambda$ is
very small or very large. An important result of these extreme-case
analyses is the asymptotic optimality of DTBLAS.

Setting the derivative of the objective function
(\ref{unconstrainedopt}) equal to zero reveals that, at the optimum
point, $\delta$, satisfies
\begin{equation}\label{setderivativezero}
e^{\lambda} (1-e^{-\delta}-\delta) + \delta = 0.
\end{equation}
Two extreme cases of large $\lambda$ and small $\lambda$ are
discussed separately in the following.

\paragraph{Large $\lambda$} In this case, solving (\ref{setderivativezero}) yields
\begin{equation}
\delta^* = 2 e^{-\lambda} + O \left( e^{-2 \lambda} \right).
\end{equation}
Consequently, $\alpha'^*$, $\tau^*$, and $\kappa^*_{_{DTBLAS}}$ are
obtained as
\begin{eqnarray}
\alpha'^*  & = & 1 -  \dfrac{1}{\lambda} + O \left( \dfrac{1}{\lambda^2} \right)\\
\tau^*_{_{DTBLAS}} & = & 1 -  \dfrac{\log(e/2)}{\lambda} + O \left(
\dfrac{1}{\lambda^2} \right)\\
\kappa^*_{_{DTBLAS}} & = & \dfrac{1}{\lambda}+O \left(
\dfrac{1}{\lambda^2} \right).
\end{eqnarray}
As it is seen from the above equations, for large values of
$\lambda$, $\delta^*$ becomes very small and $\alpha'^*$ approaches
one. This means, when large rate-per-links are demanded, it is more
crucial to manage the interference than to choose links with high
direct gain.

\paragraph{Small $\lambda$} In this case, solving (\ref{setderivativezero}) yields
\begin{equation}
\delta^* = \dfrac{1}{\lambda} + \dfrac{1}{2} + O(\lambda).
\end{equation}
Consequently, $\alpha'^*$, $\tau^*_{_{DTBLAS}}$, and
$\kappa^*_{_{DTBLAS}}$ are obtained as
\begin{eqnarray}
\alpha'^*  & = &  e^{-  \frac{1}{\lambda} - \frac{1}{2}} \left( \dfrac{1}{\lambda} + \dfrac{1}{2} + O(\lambda) \right)  \\
\tau^*_{_{DTBLAS}} & = & 1 -  \dfrac{\lambda}{2} + O \left(
\lambda^2 \right)
\label{throughputscsmallrate}\\
\kappa^*_{_{DTBLAS}} & = & \dfrac{1}{\lambda} -\dfrac{1}{2}+O \left(
\lambda \right).
\end{eqnarray}
The above equations show that for small values of $\lambda$,
$\delta^*$ is very large and $\alpha'^*$ is very small. In other
words, DTBLAS is converted to its special case, TBLAS.

The above discussion yields the following optimality result on
DTBLAS.

\begin{theorem}
Consider the rate-constrained throughput maximization problem
\eqref{constrainedproblem}. Assume $\tau^*_c$ and $\kappa^*_c$ are
the maximum achievable scaling factors of the throughput and the
number of supported links, respectively. Also, assume
$\tau_{_{DTBLAS}}^*$ and $\kappa^*_{_{DTBLAS}}$ are the maximum
scaling factor of the throughput and the number of active links when
DTBLAS is deployed. Then, a.a.s. we have
\begin{eqnarray}
\lim_{\lambda \to \infty} & (\tau^*_{_{DTBLAS}} - \tau^*_c) & =
 0,\\
\lim_{\lambda \to \infty} & (\kappa_{_{DTBLAS}}^* - \kappa^*_c) & =  0,\\
\lim_{\lambda \to 0} & (\tau^*_{_{DTBLAS}} - \tau^*_c) & =
0,\\
\lim_{\lambda \to 0} & \dfrac{\kappa_{_{DTBLAS}}^*}{\kappa^*_c} & =
1.
\end{eqnarray}
\end{theorem}

\begin{proof}
The proof of the theorem is straightforward by using the upper
bounds provided in Theorem~\ref{throughputupperbound} and the
asymptotic achievability results provided in this section.
\end{proof}

\section{Noise-Limited Regime}\label{noise-limited}

In the previous sections, we considered an interference-limited
regime in which the noise power is negligible in comparison with the
interference power. In this case, the achievable throughput is not a
function of the network \emph{SNR}. In other words, changing the
transmission powers does not affect the supportable rate of each
link. However, in a practical scenario, it is appealing to have
rates which scale by increasing $\rho$. This way, the transmission
rates can be easily adjusted by changing the transmission powers.
Specifically, it is desirable that the rate of active links a.a.s.
scale as
\begin{equation}\label{noiselimitedregime}
r_i=\log \left(1 + \dfrac{g_{ii}}{1/\rho + \beta_i} \right),\quad
\forall i \in \mathcal{A},
\end{equation}
for some $\beta_i=O(1)$, which are the design parameters. At the
same time, we require the conditions of problem
\eqref{constrainedproblem}, i.e. $r_i \geq \lambda$, be satisfied.
In this section, we show how to realize such a situation by using
DTBLAS.

According to \eqref{noiselimitedregime}, we should a.a.s. have $I_i=
\beta_i$, where $I_i$ is the interference observed by active link
$i$  and is defined in \eqref{lbound4}. However, this requires that
$\textrm{E}(I_i)=\beta_i$. Noting that
$\textrm{E}(I_i)=(k_2-1)\hat{\mu}$ (see \eqref{meaninterference}),
we conclude that all $\beta_i$s should take a same value, say
$\beta$. Hence, a necessary condition for being in the noise-limited
regime is
\begin{equation}\label{equivalentcondition2}
(k_2-1) \hat{\mu} = \beta,
\end{equation}
where $\beta=O(1)$ is a design parameters. Later, we show that
\eqref{equivalentcondition2} is also a sufficient condition for
operating in a noise-limited regime.

Based on the above discussion, we propose the following scheme for
choosing the parameters of DTBLAS for a noise-limited regime: For a
given required rate $\lambda=\log(1+\gamma_0)$ and the interference
$\beta$,
\begin{enumerate}
\item choose $\Delta$ as
\begin{equation}\label{secondcondition}
\Delta=\Delta_0=\gamma_0 (1/\rho + \beta).
\end{equation}
\item choose $\delta$ such that \eqref{equivalentcondition2} is
satisfied.
\end{enumerate}
Note that the selection of $\Delta$ is such that the rate
constraints $r_i \geq \lambda$ are satisfied. Also, as will be shown
later, the selection of $\delta$ is such that operation in the
noise-limited regime is guaranteed.

The next step is to solve \eqref{equivalentcondition2} to obtain the
value of $\delta$ and the corresponding number of active links
$k_2$. By using \eqref{conditionalmean}, which gives the value of
$\hat{\mu}$ in terms of $\delta$, it is clear that
\eqref{equivalentcondition2} holds only if $\delta \to 0 $ as $k_2
\to \infty$. In this case, \eqref{conditionalmean} converts to
$\hat{\mu}=\dfrac{\delta}{2} + O(\delta^2)$ and
\eqref{equivalentcondition2} simplifies to
\begin{equation}\label{equivalentcondition}
k_2 \delta = 2\beta \quad a.a.s
\end{equation}

To solve \eqref{equivalentcondition} and obtain $\delta$, we should
first obtain the value of $k_2$ in terms of $n$ and $\delta$. From
(\ref{kbound}) and condition (\ref{secondcondition}), the number of
links chosen by phase (i) of DTBLAS is obtained as
\begin{equation}\label{noactivelinkscphae1}
k_1 = n e ^{- \Delta_0} + O\left( \xi \sqrt{n e ^{- \Delta_0}}
\right).
\end{equation}
Also, recall from \eqref{k2concentrate} that $k_2$ is the clique
number of a random graph $\mathcal{G} (k_1,\,p) $, where $p$ is
obtained from (\ref{rgraphprob}). Since $\delta \to 0$,
(\ref{rgraphprob}) can be rewritten as
\begin{equation}\label{probability}
p=\delta ^ 2 + O(\delta ^ 3),
\end{equation}
which approaches zero as well. Note that Theorem~\ref{cliquenumber},
which was adopted from \cite{randomgraphs2}, and a similar result
that appears in \cite{randomgraphs}, are valid only for a fixed
value of $p$. A natural question is whether a similar concentration
result on the clique number of random graphs holds when $p$
approaches zero. In the following lemma, we address this issue and
obtain a concentration result on the clique number for
zero-approaching values of $p$.

\begin{lemma}\label{cliquenumber2}
Let $p=p(m)$ be such that $p=o(1)$ and $p=\omega(m^{-a})$ for all
$a>0$. For fixed $\epsilon>0$, the clique number $\textrm{cl}(G)$ of
$G \in \mathcal{G}(m,p) $ a.a.s. satisfies $ \lfloor s \rfloor \leq
\textrm{cl}(G) \leq \lfloor s \rfloor +1 $, where
\begin{equation}\nonumber
s=2 \log _b m - 2 \log _b \log _b m +1 - 4\log _b 2 -
\dfrac{\epsilon}{\log b},
\end{equation}
$b=1/p$.
\end{lemma}\vspace{6pt}

\begin{proof}
See the Appendix.
\end{proof}
By using this lemma, (\ref{noactivelinkscphae1}),
(\ref{probability}), and assuming $\xi=o(\sqrt{n})$, the number of
active links a.a.s. becomes
\begin{equation}\label{numberofactivephase2}
k_2 =\left \lfloor \dfrac{\log n - \log \log n}{-\log \delta} \right
\rfloor.
\end{equation}
Thus, \eqref{equivalentcondition} can be rewritten as
\begin{equation}\label{interferencepowerapprox}
\dfrac{\log n - \log \log n}{-\log \delta} \cdot \delta = 2 \beta.
\end{equation}
Assuming $| \log \beta | = o( \log \log n) $, it can be verified
that the solution of (\ref{interferencepowerapprox}) is
\begin{equation}\label{deltanoiselimited}
\delta= \dfrac{ 2\beta  \log \log n}{ \log n} (1+o(1)).
\end{equation}
With this value of $\delta$, the number of active links is obtained
from (\ref{numberofactivephase2}) as
\begin{equation}\label{noactivelinksnoiselimited}
k_2 = \left \lfloor \dfrac{\log n}{ \log \log n} (1+o(1)) \right
\rfloor.
\end{equation}

As mentioned before, we should show that the selected values of
$\delta$ and $\Delta$ for DTBLAS, yields the network to operate in
the noise-limited regime. The following theorem addresses this
issue.

\begin{theorem}\label{concentration3}
For the values of $\Delta$ and $\delta$ given in
\eqref{secondcondition} and \eqref{deltanoiselimited}, respectively,
the interference of active links a.a.s. satisfy
\begin{equation}
\left | I_i - \beta \right | \to 0, \quad \forall i \in \mathcal{A}.
\end{equation}
\end{theorem}

\begin{proof}
By using the central limit theorem it can be shown that
\begin{equation}
\left| I_i - \beta \right| < \dfrac{\beta \log \log n}{\sqrt{\log
n}}, \quad \forall i \in \mathcal{A},
\end{equation}
which readily yields the desired result. Since the calculations are
similar to those in the proof of Lemmas \ref{concentration1} and
\ref{concentration2}, we omit them for brevity.
\end{proof}

\begin{lemma}
Let $T_{_{NL}}$ denote the throughput achieved by DTBLAS in the
noise-limited regime described above. Then, almost surely we have
\begin{equation}
\log \left( 1 + \dfrac{\Delta_0}{1/\rho + \beta} \right) \leq
\lim_{n \to \infty} \dfrac{\log \log n}{\log n}T_{_{NL}} \leq  \log
\left( 1 + \dfrac{\Delta_0+1}{1/\rho + \beta} \right).
\end{equation}
\end{lemma}

\begin{proof}
According to Theorem~\ref{concentration3}, the throughput is
obtained as
\begin{equation}\label{throughputnoiselimited}
T_{_{NL}} = \sum_{i=1}^{k_2} \log \left(1 + \dfrac{g_{ii}}{1/\rho +
\beta} \right).
\end{equation}
Due to the fact that $g_{ii}>\Delta_0$, we have
\begin{equation}
T_{_{NL}} \geq k_2 \log \left( 1 + \dfrac{\Delta_0}{1/\rho + \beta}
\right).
\end{equation}
The left-hand-side inequality in the lemma is readily obtained by
using this inequality and the value of $k_2$ from
\eqref{noactivelinksnoiselimited}. For the right-hand-side
inequality, by utilizing the Jensen's inequality in
\eqref{throughputnoiselimited}, we obtain
\begin{equation}\label{throughputuppernl}
T_{_{NL}} \leq k_2 \log \left( 1 +
\dfrac{\frac{1}{k_2}\sum_{i=1}^{k_2} g_{ii}}{1/\rho + \beta}
\right).
\end{equation}
According to the law of large numbers and due to the fact that
$g_{ii} > \Delta_0$, we have
\begin{equation}\label{lln}
\frac{1}{k_2}\sum_{i=1}^{k_2} g_{ii} \to \textrm{E} (g_{ii} |g_{ii}
> \Delta_0 ) = 1 + \Delta_0.
\end{equation}
The result is obtained by using \eqref{throughputuppernl},
\eqref{lln}, and the value of $k_2$ from
\eqref{noactivelinksnoiselimited}.
\end{proof}

It is observed that the price for operating in the noise-limited
regime is a decrease in the throughput by a multiplicative factor of
$\log \log n$.

\section{Conclusion}\label{conclusion}

In this paper, wireless networks in Rayleigh fading environments are
studied in terms of their achievable throughput. It is assumed that
each link is either active and transmits with power $P$ and rate
$\lambda$, or remains silent. The objective is to maximize the
network throughput or equivalently the number of active links.
First, an upper bound is derived that shows the throughput and the
number of active links scale at most like $\log n$ and
$\frac{1}{\lambda} \log n$, respectively. To obtain lower bounds, we
propose two LASs (TBLAS and DTBLAS) and prove that both of them
a.a.s. yield feasible solutions for the throughput maximization
problem. In TBLAS, the activeness of each link is solely determined
by the quality of its direct channel. TBLAS, which can be
implemented in a decentralized fashion, performs very close to the
upper bound for small values of $\lambda$. However, its performance
falls below the upper bound when $\lambda$ grows large. In DTBLAS,
the mutual interference of the links are also taken into account
when choosing the active links. It is demonstrated that DTBLAS not
only performs close to the upper bound for $\lambda \to 0$, but its
performance meets the upper bound when $\lambda \to \infty$. The
above discussions take place in an interference-limited regime in
which the transmission power $P$ does not affect the transmission
rate $\lambda$. However, we show that by a proper choose of the
DTBLAS parameters, the rate-constrained network can also operate in
a noise-limited regime; this feature of the DTBLAS comes at the
price of decreasing the network throughput by a multiplicative
factor of $\log \log n$.

\appendices

\section{Proof of Lemma
\ref{concentration2}}\label{proofconcentration2}

The proof is based on the same arguments as in the proof of Lemma
\ref{concentration1}. Thus, here we just highlight the differences.

Let us define $\bar{\gamma}$ as
\begin{equation}
\bar{\gamma} =- \dfrac{(1-\alpha')\log(1-e^{-\delta})}{ \alpha'
\hat{\mu}}.
\end{equation}
Similar to the proof of Lemma~\ref{concentration1}, it is enough to
show that a.a.s.
\begin{equation}\label{equivalentbound2}
x_-' < \gamma_i < x_+',
\end{equation}
where
\begin{equation}\label{cdfparameter2}
x_{\pm}'=\bar{\gamma} \left(1 \pm c' \sqrt{\dfrac{\log \log n}{ \log
n}}(1+o ( 1 )) \right),
\end{equation}
with  $c' = c/ \bar{\gamma}$. We only prove the left side inequality
in \eqref{equivalentbound2}; the other inequality can be proved in a
similar manner.

Let $\mathcal{L}_3$ denote the event that
\begin{equation}\label{SINRlowerbound2}
\gamma_i > x_-',  \quad \forall  i \in \mathcal{A},
\end{equation}
In the following, we show that $\textrm{P}(\mathcal{L}_3) \to 1$ for
some $c'>0$.

Note that with $\Delta=(1-\alpha') \log n$, the parameter $k_+'$ in
Corollary \ref{concentrationactive2} is obtained as
\begin{eqnarray}
k_+' & = & \kappa_{_{DTBLAS}} \log n - a \log \log n\\
& < & \kappa_{_{DTBLAS}} \log n,\label{upperboundk+}
\end{eqnarray}
where $\kappa_{_{DTBLAS}}$ is given in \eqref{noactivelinksc} and
$a>0$ is a constant. Denoting the cdf of $\gamma_i$ conditioned on
$|\mathcal{A}|=k$ by $F_{\gamma}(x,k)$, we have
\begin{eqnarray}
\textrm{P} (\mathcal{L}_3) & \stackrel{(a)}{>} & \left(1-
F_{\gamma}(x_-',k_+') \right)^{k_+'} \textrm{P} \left( k_-' \leq
|\mathcal{A}| \leq k_+' \right)\\
& \stackrel{(b)}{\approx} & \left(1- F_{\gamma}(x_-',k_+') \right)^{k_+'}\\
& \stackrel{(c)}{>} & \left(1- F_{\gamma}(x_-',\kappa_{_{DTBLAS}}
\log n) \right)^{ \kappa_{_{DTBLAS}} \log n},\label{khastehshodam2}
\end{eqnarray}
where (a) is obtained in the same manner as \eqref{khastehshodam},
(b) results from Corollary \ref{concentrationactive2}, and (c) is
due to  \eqref{upperboundk+} and the fact that
$(1-F_{\gamma}(x,k))^k $ is a decreasing function of $k$. To show
that the RHS of \eqref{khastehshodam2} tends to one, we upper bound
$ \kappa_{_{DTBLAS}} \log n F_{\gamma}(x_-',\kappa_{_{DTBLAS}} \log
n) $ and show that it approaches zero.

Similar to the derivation of \eqref{upperFGamma}, it can be shown
that
\begin{equation}\label{upperFGamma2}
F_{\gamma}(x_-',\kappa_{_{DTBLAS}} \log n) <  \textrm{P}
\left(\sum_{\substack{j=1\\ j \neq i}} ^ {\kappa_{_{DTBLAS}} \log n
}g_{ji} \geq \dfrac{\Delta}{x_-'} -\frac{1}{\rho} \right).
\end{equation}
Let us define $Y_j = g_{ji}- \hat{\mu}$, where $\hat{\mu}$ is
obtained from (\ref{conditionalmean}). Random variable $Y_j$ has the
variance $\nu=\hat{\sigma}^2$, where $\hat{\sigma}^2$ is given in
\eqref{conditionalvariance}. By defining $Z=\dfrac{1}{\sqrt{\nu
(\kappa_{_{DTBLAS}} \log n -1)}} \sum_{\substack{j=1\\ j \neq i}} ^
{\kappa_{_{DTBLAS}} \log n } Y_j$, \eqref{upperFGamma2} can be
reformulated as
\begin{equation}\label{upperFGamma2reformulated}
F_{\gamma}(x_-',\kappa_{_{DTBLAS}} \log n) < 1 - \textrm{P} (Z \leq
y),
\end{equation}
where
\begin{equation}\label{defy2}
y  =  \dfrac{\frac{\Delta}{x_-'} -\frac{1}{\rho} -
(\kappa_{_{DTBLAS}} \log n-1)\hat{\mu}}{\sqrt{(\kappa_{_{DTBLAS}}
\log n-1)\hat{\sigma}^2}}.
\end{equation}
By substituting $\Delta=(1-\alpha') \log n$ and the value of $x_-'$
from \eqref{cdfparameter2} into \eqref{defy2}, we obtain
\begin{eqnarray}
y & = &  c' \sqrt{\dfrac{\kappa_{_{DTBLAS}}
\hat{\mu}^2}{\hat{\sigma}^2}}\sqrt{\log \log n} \left( 1 + O \left(
\dfrac{1}{\sqrt{\log n \log \log n }} \right)
\right).\label{valuey2}
\end{eqnarray}

It is straightforward to show that the moment-generating function of
$Y_j$ exists around zero. Hence, the Cram\'er's condition is
satisfied. Also, by choosing $m = \kappa_{_{DTBLAS}} \log n -1$, the
condition $y = O (m^{1/6})$ is satisfied, as well. As a result,
Theorem~\ref{clt} can be utilized to calculate the
RHS\eqref{upperFGamma2reformulated} as
\begin{equation}\label{rhs3}
1 - \textrm{P} (Z \leq y) = [1-\Phi(y)]\exp\left(\dfrac{\theta_3
y^3}{6
 \sqrt{\nu^3 \kappa_{_{DTBLAS}} \log n}}\right) +O\left(
\dfrac{e^{-\frac{y^2}{2}}}{\sqrt{\kappa_{_{DTBLAS}} \log n}} \right)
\end{equation}

By combining \eqref{upperFGamma2reformulated}, \eqref{rhs3}, and
\eqref{valuey2}, and noting that $\theta_3$ is a constant,
$y^3=o(\sqrt{\kappa_{_{DTBLAS}} \log n})$, and $1-\Phi(y) <
\frac{e^{-y^2/2}}{y}$, we conclude that
\begin{equation}
\kappa_{_{DTBLAS}} \log n F_{\gamma}(x,\kappa_{_{DTBLAS}} \log n) <
 \kappa_{_{DTBLAS}} \log n \dfrac{e^{-\frac{y^2}{2}}}{y}
\end{equation}
\begin{equation}\nonumber
=  \exp\left( (1-\frac{c'^2 \kappa_{_{DTBLAS}}
\hat{\mu}^2}{2\hat{\sigma}^2}) \log \log n  +O(\log \log \log n)
\right)
\end{equation}
It is clear that if $c'$ is chosen large enough, the above upper
bound approaches zero as $n \to \infty$. This completes the proof.

\section{Proof of Lemma \ref{cliquenumber2}}

The proof is based on the standard second moment method.

\subsection{Preliminary Calculations}

Assume $Y_s$ is the number of cliques of size $s$ in $G$. Let us
denote its mean and variance by $\mu_s$ and $\sigma_s^2$,
respectively. According to \cite{randomgraphs}, we have
\begin{eqnarray}
\label{exprcliques} \mu_s & = & \binom{m}{s} p ^{\binom{s}{2}},\\
\label{ratiorclique} \dfrac{\sigma_s^2}{\mu_s^2} & = &
\sum_{\ell=2}^s \dfrac{\binom{s}{\ell}
\binom{m-s}{s-\ell}}{\binom{m}{s}} (b^{\binom{\ell}{2}}-1),
\end{eqnarray}
where $b=1/p$. By applying the Stirling's approximation to
(\ref{exprcliques}), we obtain
\begin{eqnarray}\label{exprcliqueup1}
\mu_s & = & \dfrac{m^{m+\frac{1}{2}}}{ \sqrt{2 \pi}
s^{s+\frac{1}{2}} (m-s)^{m-s+\frac{1}{2}}   } p ^
{\frac{s(s-1)}{2}}\\
& \leq & \dfrac{1}{ \left( \frac{s}{m}  \right) ^ s   \left(
1-\frac{s}{m}  \right)^m } p ^ {\frac{s(s-1)}{2}}
\end{eqnarray}
For any $\epsilon>0$, the inequality $1-x \geq e^{-(1+\epsilon)x}$
holds for sufficiently small values of $x$. Since we are interested
in small values of $s/m$, from this inequality and
(\ref{exprcliqueup1}), we obtain
\begin{equation}\label{exprcliqueup2}
\mu_s \leq e^ {s ( \log m - \log s +(1+\epsilon) - \frac{s-1}{2}
\log b )}
\end{equation}
Equation (\ref{ratiorclique}) is readily converted to the following
inequality
\begin{equation}\label{ratiorcliqueup1}
\dfrac{\sigma_s^2}{\mu_s^2}  \leq  \sum_{\ell=2}^s F_{\ell},
\end{equation}
where
\begin{equation}
F_{\ell}= \dfrac{\binom{s}{\ell} \binom{m-s}{s-\ell}}{\binom{m}{s}}
b^{\binom{\ell}{2}}.
\end{equation}
By using the definition of the binomial coefficients, we obtain
\begin{eqnarray}
F_{\ell} & \leq & 2^s \cdot \frac{(m-s)!}{m!} \cdot
\frac{(m-s)!}{(m-2s+\ell)!} \cdot \frac{s!}{(s-\ell)!} \cdot
b^{\frac{\ell (\ell -1)}{2}}\\
& \leq & \dfrac{2^s \cdot (m-s)^{s-\ell} \cdot s^{\ell} }{ (m-s)^{s}
} \cdot b^{\frac{\ell (\ell -1)}{2}}\\
& = & 2^s \cdot \left( \frac{m}{s} -1 \right) ^ {- \ell} \cdot
b^{\frac{\ell (\ell -1)}{2}}
\end{eqnarray}
Noting that $\frac{m}{s} \gg 1$, the above inequality can be
approximately written as
\begin{equation}\label{Flup}
F_{\ell} \leq 2^s \cdot \left( \frac{s}{m} \right) ^ { \ell} \cdot
b^{\frac{\ell (\ell -1)}{2}}.
\end{equation}
Using (\ref{ratiorcliqueup1}) and (\ref{Flup}), we obtain
\begin{equation}\label{ratiorcliqueup2}
\dfrac{\sigma_s^2}{\mu_s^2}  \leq  \sum_{\ell=2}^s e^{g(\ell)},
\end{equation}
where
\begin{equation}
g(\ell)=s \log 2 + \ell ( \log s - \log m + \frac{\ell}{2} \log b -
\frac{1}{2} \log b )
\end{equation}
is a quadratic convex function with a minimum at $\ell _0 =
\frac{\log m}{\log b} - \frac{\log s}{\log b} + \frac{1}{2} $.
Define
\begin{equation}
s_0=2 \log _b m - 2 \log _b \log _b m  - 2\log _b 2.
\end{equation}
It is easy to show that if $s>s_0$, then $g(s)>g(2)$. Hence,
(\ref{ratiorcliqueup2}) can be simplified as
\begin{equation}\label{ratiorcliqueup3}
\dfrac{\sigma_s^2}{\mu_s^2}  \leq  e^{\log s + g(s)}.
\end{equation}

\subsection{Proof}

According to the Markov's inequality, we have
\begin{equation}\label{markov}
\textrm{P} \left\{ Y_s \geq 1  \right\} \leq \mu_s.
\end{equation}
For a fixed $\epsilon>0$, define
\begin{equation}
s_1=2 \log _b m - 2 \log _b \log _b m +1 +2\log _b (e/2) +
\dfrac{\epsilon}{\log b}.
\end{equation}
Using (\ref{exprcliqueup2}), it is easy to verify that for $s \geq
s_1$, we have $\mu_s \to 0$ as $ m \to \infty $. Hence, from
(\ref{markov}), we conclude that
\begin{equation}
\textrm{P} \left\{ Y_{s} \geq 1  \right\}  \to 0, \quad \textrm{for}
\quad s \geq s_1
\end{equation}
as $m \to \infty$. This means a.a.s. the clique number of $G$ is
less than $s_1$, i.e., we have the following upper bound on
$\textrm{cl}(G)$
\begin{equation}\label{upperboundcliquenumber}
\textrm{cl}(G) < s_1 \quad a.a.s.
\end{equation}
According to the Chebyshev's inequality, we have
\begin{equation}\label{chebyshev}
\textrm{P} \left\{ Y_s = 0  \right\} \leq
\dfrac{\sigma_s^2}{\mu_s^2}.
\end{equation}
For a fixed $\epsilon>0$, define
\begin{equation}
s_2=2 \log _b m - 2 \log _b \log _b m +1 - 4\log _b 2 -
\dfrac{\epsilon}{\log b}.
\end{equation}
Using (\ref{ratiorcliqueup3}), it is easy to verify that for $s \leq
s_2$, we have $\sigma_s^2 / \mu_s^2 \to 0$ as $ m \to \infty$.
Hence, from (\ref{chebyshev}), we conclude that
\begin{equation}
\textrm{P} \left\{ Y_{s} = 0  \right\}  \to 0, \quad \textrm{for}
\quad s \leq s_2
\end{equation}
as $m \to \infty$. This means a.a.s. the clique number of $G$ is not
less than $\lfloor s_2 \rfloor$, i.e., we have the following lower
bound on $\textrm{cl}(G)$
\begin{equation}\label{lowerboundcliquenumber}
\textrm{cl}(G) \geq \lfloor s_2 \rfloor  \quad a.a.s.
\end{equation}
For sufficiently small $\epsilon$, the difference between the upper
bound $s_1$ and the lower bound $s_2$ is less than one. Hence, from
(\ref{upperboundcliquenumber}) and (\ref{lowerboundcliquenumber}) we
can conclude that
\begin{equation}
\lfloor s_2 \rfloor \leq \textrm{cl}(G) \leq \lfloor s_2 \rfloor +1
\quad a.a.s.
\end{equation}

\section*{Acknowledgment}

The authors would like to thank T. Luczak and S. Oveis Gharan for
helpful discussions on the proof of Lemma~\ref{cliquenumber2}.


\end{document}